\documentclass[12pt]{article}
\usepackage{amsmath, amsthm}
\usepackage{amssymb}
\usepackage{amsfonts}

\newtheoremstyle{theorem}
  {10pt}
  {10pt}
  {\sl}
  {\parindent}
  {\bf}
  {. }
  { }
  {}
\theoremstyle{theorem}
\newtheorem{theorem}{Theorem}
\newtheorem{corollary}[theorem]{Corollary}
\newtheorem{example}[theorem]{Example}
\newtheoremstyle{defi}
  {10pt}
  {10pt}
  {\rm}
  {\parindent}
  {\bf}
  {. }
  { }
  {}
\theoremstyle{defi}
\newtheorem{definition}[theorem]{Definition}

\input{tcilatex}
\begin{document}

\title{L\'{e}vy driven models and derivative pricing}
\author{Alexander Kushpel and Jeremy Levesley \\
Department of Mathematics, University of Leicester\\
Leicester, LE1 7RH, UK\\
ak412@le.ac.uk and jl1@le.ac.uk\\
[2pt] }
\maketitle

\begin{abstract}
In this article we develop a general method for derivative pricing which is
based on results obtained in \cite{klf}. This approach has its roots in
Shannon's Information Theory. The notion of $\lambda$-analyticity of L\'{e}%
vy models is introduced on the basis of which new representations of the
pricing integral are obtained. It is shown that popular in applications L%
\'{e}vy models are $\lambda$-analytic. We apply these results to derive a
general algorithm for pricing of European call options.


\end{abstract}


\section{Introduction}

\label{sec1} 

Consider a frictionless market consisting of a riskless bond and stock which
is modeled by an exponential L\'{e}vy process $S_{t}=S_{0}\exp (X_{t})$
under a fixed equivalent martingale measure $\mathbb{Q}$ with a given
constant riskless rate $r>0$. For such market consider a contract (European
call option) which gives to its owner the right but not the obligation to
buy the underlying asset for the fixed price $K$ at the fixed expiry date $T$%
. We need to evaluate its price $F_{call}$. In this case the payoff function
has the form $F(x)=(S_{0}e^{x}-K)_{+}$, where $(a)_{+}:=\max \{a,0\}$, $K$
is the strike price and $x=\ln (S_{t})$.

For a finite measure $\mu $ on $\mathbb{R}$ define its formal Fourier
transform $\mathbf{F}$ and its formal inverse $\mathbf{F}^{-1}$ as 
\begin{equation*}
\mathbf{F}\mu (y)=\int_{\mathbb{R}}e^{-ixy}d\mu (x),\,\,\,\mathbf{F}^{-1}\mu
(x)=\frac{1}{2\pi }\int_{\mathbb{R}}e^{ixy}d\mu (y).
\end{equation*}%
Remind that any L\'{e}vy process $X=\{X_{t}\}_{t\in \mathbb{R}_{+}}$ is
uniquely determined by its characteristic exponent $\psi (x)$ which is
defined as $\mathbb{E}\left[ e^{ixX_{t}}\right] =e^{-t\psi (x)}$, $x\in 
\mathbb{R}$, $t\in \mathbb{R}_{+}$. Let $\chi _{\lbrack -1,1]}$ be the
characteristic function of $[-1,1]$, $a\geq 0$ and $b\in \mathbb{R}$ then
any characteristic exponent $\psi $ of any jump-diffusion process $%
X=\{X_{t}\}_{t\in \mathbb{R}_{+}}$ admits L\'{e}vy-Khintchine's
representation 
\begin{equation}
\psi (y)=-\frac{1}{2}ay^{2}-iby-\int_{\mathbb{R}}\left( 1-e^{iyx}+iy\chi
_{\lbrack -1,1]}(x)\right) \Pi (dx),  \label{2-kh}
\end{equation}%
where $\Pi :\mathbb{R}\rightarrow \mathbb{R}$ such that 
\begin{equation}
\int_{\mathbb{R}}\min \{1,x^{2}\}\Pi (dx)<\infty ,\,\,\,\Pi \left(
\{0\}\right) =0.  \label{333}
\end{equation}%
The corresponding density function $p_{t}$ can be written as 
\begin{equation*}
p_{t}(y)=\frac{1}{2\pi }\int_{\mathbb{R}}e^{iy\xi }\mathbb{E}[e^{i\xi
X_{t}}]d\xi =\mathbf{F}^{-1}\left( \exp (-t\psi (\cdot )\right) (y)
\end{equation*}%
\begin{equation}
=\frac{1}{2\pi }\int_{\mathbb{R}}\exp \left( iy\xi -t\int_{\mathbb{R}%
}(1-e^{i\xi x}+i\xi x\chi _{\lbrack -1,1]}(x))\Pi (dx)\right) d\xi .
\label{4-dens}
\end{equation}%
See \cite{McKean}, \cite{sato} for more information.


\section{The results}

\label{sec2}

Let $F(\cdot )$ be a reward function then the no-arbitrage price $V$ of the
claim with the terminal payoff $F(S_{0}e^{X_{T}})$ is the expectation of the
discounted terminal payoff $e^{-rT}F(S_{0}e^{X_{T}})$ given that $X_{0}=0$,
i.e., 
\begin{equation*}
V=\mathbb{E}^{\mathbb{Q}}\left[ e^{-rT}F(X_{T})|X_{0}=0\right]
\end{equation*}%
\begin{equation}
=e^{-rT}\mathbb{E}^{\mathbb{Q}}\left[ F\left( S_{0}e^{X_{T}}\right) \right]
=e^{-rT}\int_{\mathbb{R}}p_{T}^{\mathbb{Q}}(y)F\left( S_{0}e^{y}\right) dy,
\label{price}
\end{equation}%
where $p_{T}^{\mathbb{Q}}(y)$ is the density function generated by the
chosen equivalent martingale measure $\mathbb{Q}$, i.e. such measure that
for the discounted process $\widetilde{S}%
_{t}=e^{-rt}S_{t}=e^{-rt}S_{0}e^{X_{t}}$ and a fixed filtration $\{\mathcal{F%
}_{l}\}_{l\in \mathbb{R}_{+}}$ the martingale condition holds, 
\begin{equation*}
\widetilde{S}_{t}=\mathbb{E}^{\mathbb{Q}}\left[ f\left( \widetilde{S}_{t}|%
\mathcal{F}_{l}\right) \right] ,\,\,\,\forall 0\leq l<t\leq T,
\end{equation*}%
where $T>0$ is the maturity time. In particular, let $F(\zeta )=(\zeta
-K)_{+} $ then applying (\ref{price}) the price of the European call option
can be formally given by 
\begin{equation}
F_{call}=e^{-rT}K\int_{\ln (K/S_{0})}^{\infty }p_{T}(y)\left( e^{\ln
(S_{0}/K)+y}-1\right) dy.  \label{integral2}
\end{equation}%
In applications it is important to construct such pricing theory which
includes customary used reward functions $F$, for instance European call
payoff $F(S_{0}e^{x})=(S_{0}e^{x}-K)_{+}$. This kind of payoff has an
exponential grows as $x\rightarrow \infty $. Hence the integral (\ref%
{integral2}) can be understood just in the sense of generalized functions.
To guarantee the existence of the integral (\ref{integral2}) we need to
assume that the function $e^{-t\psi (z)}$ admits an analytic extension onto
the strip $\{z|0\leq \Im z\leq a,a>1\}$.

For a fixed $R>0$ consider two piecewise smooth curves 
\begin{equation*}
\lambda _{+}(\theta )=f(\theta )+i(\alpha _{+}+a_{+}(\theta
)):[-R,R]\rightarrow \{z|\Im z>0\},\,\,\alpha _{+}>0
\end{equation*}%
\begin{equation*}
\lambda _{-}(\theta )=g(\theta )+(i\alpha _{-}+a_{-}(\theta
)):[-R,R]\rightarrow \{z|\Im z<0\},\,\,\alpha _{-}<0,
\end{equation*}%
where $f(\theta )\leq 0$ if $\theta \leq 0$ and $f(\theta )\geq 0$ if $%
\theta \geq 0$, $g(\theta )\leq 0$ if $\theta \leq 0$ and $g(\theta )\geq 0$
if $\theta \geq 0$, $a_{+}(\theta )\geq 0$ is an increasing function on $%
[0,R]$ and decreasing on $[-R,0]$, $a_{-}(\theta )\leq 0$ is an increasing
function on $[-R,0]$ and decreasing on $[0,R]$. Put 
\begin{equation*}
\varrho _{+}:=\left\vert \lambda _{+}\left( R\right) \right\vert ,\varrho
_{-}:=\left\vert \lambda _{-}\left( R\right) \right\vert .
\end{equation*}%
Assume that 
\begin{equation*}
\lim_{R\rightarrow \infty }\varrho _{+}=\infty ,\,\,\,\lim_{R\rightarrow
\infty }\varrho _{-}=\infty .
\end{equation*}%
Consider six contours 
\begin{equation*}
\gamma _{1}(R):=\left\{ z|z=\varrho _{+}e^{i\phi },\,\phi \in \lbrack 
\mathrm{arg}(f(R)+i(\alpha _{+}+a_{+}(R))),0]\right\} ,
\end{equation*}%
\begin{equation*}
\gamma _{2}(R):=[\varrho _{+},\varrho _{-}],
\end{equation*}%
\begin{equation*}
\gamma _{3}(R):=\left\{ z|z=\varrho _{-}e^{i\phi },\,\phi \in \lbrack 0,%
\mathrm{arg}(g(R)+i(\alpha _{-}+a_{-}(R)))]\right\} ,
\end{equation*}%
\begin{equation*}
\gamma _{4}(R):=\left\{ z|z=\varrho _{+}e^{i\phi },\,\phi \in \lbrack 
\mathrm{arg}(g(-R)+i(\alpha _{+}+a_{+}(-R))),\pi ]\right\}
\end{equation*}%
\begin{equation*}
\gamma _{5}(R):=[-\varrho _{+},-\varrho _{-}],
\end{equation*}%
\begin{equation*}
\gamma _{6}(R):=\left\{ z|z=\varrho _{+}e^{i\phi },\,\phi \in \lbrack \pi ,%
\mathrm{arg}(g(-R)+i(\alpha _{+}+a_{+}(-R)))]\right\} .
\end{equation*}

\begin{definition}
\label{d1} We say that a L\'{e}vy process $X=\{X_{t}\}_{t\in \mathbb{R}_{+}}$
is $(\lambda _{-},\lambda _{+})$-analytic if for any $R>0$ its
characteristic exponent $\psi (z)$ admits analytic extension into the domain 
$\Omega _{R}\ni 0$ bounded by $\lambda _{+}\left( \cdot \right) \cup \lambda
_{+}\left( \cdot \right) \cup \bigcup_{k=1}^{6}\gamma _{k}(R)$ and 
\begin{equation*}
\lim_{R\rightarrow \infty }\int_{\gamma _{k}(R)}e^{iyz-\tau \psi
(z)}dz=0,\,\,1\leq k\leq 6,\tau >0.
\end{equation*}
\end{definition}

Next statement gives a useful representation of the density function $%
p_{\tau}(y)$ given by (\ref{4-dens}).

\begin{theorem}
\label{representation} Let $X=\{X_{t}\}_{t \in \mathbb{R}_{+}}$ be a $%
(\lambda_{-},\lambda_{+})$-analytic process then 
\begin{equation*}
p_{\tau}(y)=\frac{1}{2\pi(e^{\alpha_{+}y}+e^{\alpha_{-}y})} \int_{\mathbb{R}%
} \left(e^{iy(f(\theta)+ia_{+}(\theta))}N(\theta)
+e^{iy(g(\theta)+ia_{-}(\theta))}M(\theta)\right) d\theta,
\end{equation*}
where 
\begin{equation*}
N(\theta):=e^{-\tau \psi(f(\theta)+i\alpha_{+} +ia_{+}(\theta))}(\dot{f}%
(\theta)+i\dot{a}_{+}(\theta))
\end{equation*}
and 
\begin{equation*}
M(\theta):=e^{-\tau \psi(g(\theta)+i\alpha_{-} +ia_{-}(\theta))} (\dot{g}%
(\theta)+i\dot{a}_{-}(\theta)).
\end{equation*}
\end{theorem}

\begin{proof}
Let $\gamma_{7}(R):=\{z|z=\lambda_{+}(\theta),\theta \in [-R,R]\}$.
Since the process $X=\{X_{t}\}_{t \in \mathbb{R}}$ is $(\lambda_{-},\lambda_{+})$-analytic then using Cauchy's theorem we get
\[
\oint_{\gamma_{1} \cup [\varrho_{+},-\varrho_{+}] \cup \gamma_{6} \cup \gamma_{7}} e^{iyz-\tau \psi(z)}dz=0.
\]
Applying $(\lambda_{-},\lambda_{+})$-analyticity and letting $R \rightarrow \infty$ we get
\[
p_{\tau}(y)=
\frac{1}{2\pi} \int_{\mathbb{R}} e^{iy\xi - \tau \psi(\xi)}d\xi
=\frac{1}{2\pi}\,\lim_{R \rightarrow \infty}\,
\int_{[-R,R]} e^{iy\xi - \tau \psi(\xi)}d\xi
\]
\[
=
\frac{e^{-\alpha_{+}y}}{2\pi} \int_{\mathbb{R}} e^{iyf(\theta)-ya_{+}(\theta)-\tau \psi(f(\theta)+i\alpha_{+}+ia_{+}(\theta))}
(\dot{f}(\theta)+i\dot{a}_{+}(\theta))d\theta.
\]
Similarly,
\[
p_{\tau}(y)=
\frac{e^{-\alpha_{-}y}}{2\pi} \int_{\mathbb{R}} e^{iyf(\theta)-ya_{-}(\theta)-\tau \psi(g(\theta)+i\alpha_{-}+ia_{-}(\theta))}
(\dot{g}(\theta)+i\dot{a}_{-}(\theta))d\theta.
\]
Hence
\[
p_{\tau}(y)=\frac{1}{2\pi(e^{\alpha_{+}y}+e^{\alpha_{-}y})}
\]
\[
\times \int_{\mathbb{R}} (e^{iy(f(\theta)+ia_{+}(\theta))}
e^{-\tau \psi(f(\theta)+i\alpha_{+} +ia_{+}(\theta))}(\dot{f}(\theta)+i\dot{a}_{+}(\theta))
\]
\[
+e^{iy(g(\theta)+ia_{-}(\theta))}e^{-\tau \psi(g(\theta)+i\alpha_{-} +ia_{-}(\theta))}
(\dot{g}(\theta)+i\dot{a}_{-}(\theta))
d\theta.
\]
\end{proof}

\begin{corollary}
\label{corollary 1} Let $X=\{X_{t}\}_{t \in \mathbb{R}_{+}}$ be a $%
(0,\lambda_{+})$-analytic process, $\lambda_{+}(\theta)=f(\theta)+i(%
\alpha_{+}+a_{+}(\theta))$ then 
\begin{equation*}
p_{\tau}(y)=\frac{e^{-\alpha_{+}y}}{2\pi} \int_{\mathbb{R}%
}e^{iyf(\theta)-ya_{+}(\theta)-\tau
\psi(f(\theta)+i\alpha_{+}+ia_{+}(\theta))} (\dot{f}(\theta)+i\dot{a}%
_{+}(\theta))d\theta.
\end{equation*}
\end{corollary}

We shall use Wiener spaces $W_{\sigma }(\mathbb{R})\subset L_{2}(\mathbb{R})$%
, i.e. entire functions of exponential type $\sigma >0$ whose Fourier
transform has support on $[-\sigma ,\sigma ]$. An important property of
Wiener spaces is given by the following generalization of the
Whittaker-Kotel'nikov-Shannon formula \cite{klf}.

\begin{theorem}
\label{whitteker} Let $\lambda_{\sigma} \in L_{2}(\mathbb{R})$ and $%
\lambda_{\sigma}(y)=1$ if $y \in [-\sigma,\sigma]$, $\lambda_{\sigma}(y)=0$
if $y \in \mathbb{R}\setminus [-2\sigma,2\sigma]$ and 
\begin{equation*}
J_{m,\lambda_{\sigma}}(x)=\frac{1}{2\sigma}\left(\mathbf{F}
\lambda_{\sigma}\right)\left(-x+\frac{\pi m}{\sigma}\right).
\end{equation*}
Then 
\begin{equation*}
f(x)=\sum_{m \in \mathbb{Z}} f\left(\frac{\pi m}{\sigma}\right)J_{m,%
\lambda_{\sigma}}(x)
\end{equation*}
for any $f \in W_{\sigma}(\mathbb{R})$.
\end{theorem}

Applying Theorems \ref{representation} and \ref{whitteker} we get the
following result.

\begin{theorem}
\label{approximant} Let $X=\{X_{t}\}_{t \in \mathbb{R}_{+}}$ be a $%
(\lambda_{-},\lambda_{+})$-analytic process. Then the approximant $%
p_{\tau}^{\ast}(y)$ for the density function $p_{\tau}(y)$ is given by 
\begin{equation*}
p_{\tau}^{\ast}(y)=\frac{\chi_{[-\sigma,\sigma]}(y)}{2(e^{\alpha_{+}y}+e^{%
\alpha_{-}y})} \sum_{k \in \mathbb{Z}}
\left(e^{iy(f(\theta)-\theta)-ya_{+}(\pi k/\sigma)} N\left(\frac{\pi k}{%
\sigma}\right) \right.
\end{equation*}
\begin{equation*}
\left. +e^{iy(g(\theta)-\theta)-ya_{-}(\pi k/\sigma)} M\left(\frac{\pi k}{%
\sigma}\right)\right) e^{i\pi k y/\sigma}.
\end{equation*}
\end{theorem}

\begin{proof}
Observe that
\[
e^{iyf(\theta)-ya_{+}(\theta)}N(\theta)=e^{iy\theta}e^{iy(f(\theta)-\theta)-ya_{+}(\theta)}N(\theta)
\]
and
\[
e^{iyg(\theta)-ya_{-}(\theta)}M(\theta)=e^{iy\theta}e^{iy(f(\theta)-\theta)-ya_{-}(\theta)}M(\theta).
\]
Let
\[
N^{\ast}(\theta):=e^{iy(f(\theta)-\theta)-ya_{+}(\theta)}N(\theta)
\]
and
\[
M^{\ast}(\theta):=e^{iy(g(\theta)-\theta)-ya_{-}(\theta)}M(\theta).
\]
Then
\[
p_{\tau}(y)=\frac{1}{2\pi(e^{\alpha_{+}y}+e^{\alpha_{-}y})}\int_{\mathbb{R}} e^{iy\theta}\left(N^{\ast}(\theta)+M^{\ast}(\theta)\right)d\theta.
\]
Applying Theorem \ref{whitteker} for a fixed $\sigma>0$ we get
\[
p_{\tau}^{\ast}(y)=
\frac{1}{2\sigma(e^{\alpha_{+}y}+e^{\alpha_{-}y})}
\sum_{k \in \mathbb{Z}}
\left(N^{\ast}\left(\frac{\pi k}{\sigma}\right)+
M^{\ast}\left(\frac{\pi k}{\sigma}\right)\right)
\]
\[
\times \frac{1}{2\pi} \int_{\mathbb{R}} e^{iy\theta} J_{k, \lambda_{\sigma}}(\theta)d\theta.
\]
Put $\lambda_{\sigma}(\theta)=\chi_{[-\sigma, \sigma]}(\theta)$.
Since $J_{k,\chi_{[-\sigma,\sigma]}} \in L_{2}(\mathbb{R})$, $\forall k \in \mathbb{Z}$ then from the Plancherel's Theorem we obtain
\[
\frac{1}{2\pi} \int_{\mathbb{R}}
e^{iy\theta} J_{k,\chi_{[-\sigma,\sigma]}}(\theta)d\theta
={\bf F}^{-1} \circ \left({\bf F} \chi_{[-\sigma,\sigma]}\right)
\left(-x+\frac{\pi k}{\sigma}\right)(\theta)
\]
\[
=\chi_{[-\sigma,\sigma]}(y) e^{i\pi k y/\sigma}.
\]
\end{proof}
From Theorem \ref{approximant} and Corollary \ref{corollary 1} we get

\begin{corollary}
\label{corollary2} Let $X=\{X_{t}\}_{t\in \mathbb{R}_{+}}$ be a $(0,\lambda
_{+})$-analytic process and $f(\theta )=\theta $ then 
\begin{equation*}
p_{\tau }^{\ast }(y)=\frac{e^{-\alpha _{+}y}\chi _{\lbrack -\sigma ,\sigma
]}(y)}{2\sigma }\sum_{k\in \mathbb{Z}}e^{-ya_{+}(\pi k/\sigma )}N\left( 
\frac{\pi k}{\sigma }\right) e^{i\pi ky/\sigma }.
\end{equation*}%
In particular, if $\lambda _{+}(\theta )=\theta +i\alpha _{+}$, $\theta \in 
\mathbb{R}$ then 
\begin{equation*}
p_{\tau }^{\ast }(y)=\frac{e^{-\alpha _{+}y}\chi _{\lbrack -\sigma ,\sigma
]}(y)}{2\sigma }\sum_{k\in \mathbb{Z}}e^{i\pi ky/\sigma -\tau \psi (\theta
+i\alpha _{+})}.
\end{equation*}
\end{corollary}

We shall consider here various applications of our general results. A L\'{e}%
vy process is called a KoBoL process of order $0<\nu <2$ if it is a purely
discontinuous (i.e. $a=b=0$ in (\ref{2-kh})) with the L\'{e}vy measure of
the form $c_{+}\Pi ^{+}(dx)+c_{-}\Pi ^{-}(dx)$, where 
\begin{equation*}
\Pi ^{+}(dx)=\left( \max \{x,0\}\right) ^{-\nu -1}e^{-\lambda _{+}x}dx,
\end{equation*}%
\begin{equation*}
\Pi ^{-}(dx)=\left( -\min \{x,0\}\right) ^{-\nu -1}e^{-\lambda _{-}x}dx,
\end{equation*}%
$c_{+},c_{-}>0$ and $\lambda _{-}<0<\lambda _{+}$. It is easy to check that
the condition (\ref{333}) is satisfied. Using integration by parts in (\ref%
{2-kh}) it is possible to show that the corresponding characteristic
exponent $\psi _{\ast }(\xi )$ has the form (see \cite{bl1}) 
\begin{equation*}
\psi _{\ast }(\xi )=-i\mu \xi +c_{+}\Gamma (-\nu )((-\lambda _{-})^{\nu
}-(-\lambda _{-}-i\xi )^{\nu })
\end{equation*}%
\begin{equation*}
+c_{-}\Gamma (-\nu )(\lambda _{+}^{\nu }-(\lambda _{+}+i\xi )^{\nu }),
\end{equation*}%
where $\nu \in (0,2)\setminus \{1\}$ and $\mu \in \mathbb{R}$.

\begin{theorem}
\label{positive} Any KoBoL exponent $\psi _{\ast }$ with parameters $\mu
\geq 0$, $c_{+}=c_{-}=c>0$ and $\nu \in (0,1/2]$ is $(0,\lambda _{+})$%
-analytic L\'{e}vy process, where $\lambda _{+}(z)\in \{z|\Im
z>0,\,\,z\notin \lbrack i\lambda _{+},i\infty )\}$.
\end{theorem}

\begin{proof}
Clearly, $\psi_{\ast}(\xi)$ is analytic in the domain $\mathbb{C}\setminus \{[i\lambda_{+}, i\infty) \cup [i\lambda_{-}, -i\infty)\}$. Hence it is sufficient to show that  $\lim_{R \rightarrow \infty} I_{+}(R,y,\tau)=0$ and
$\lim_{R \rightarrow \infty} I_{-}(R,y,\tau)=0$
for any $y \geq 0$ (see (\ref{integral2})) and $\tau >0$, where
\begin{equation} \label{int1}
I_{+}(R)=I_{+}(R,y,\tau):= \left|\int_{\xi \in \Gamma_{R}^{+}}
e^{iy\xi} e^{-\tau \psi(\xi)}d\xi \right|,
\end{equation}
\begin{equation} \label{int2}
I_{-}(R)=I_{-}(R,y,\tau):= \left|\int_{\xi \in \Gamma_{R}^{-}}
e^{iy\xi} e^{-\tau \psi(\xi)}d\xi \right|,
\end{equation}
\[
\Gamma_{R}^{+}=\{\xi|\,\xi=\varrho_{+}e^{i\phi}\},\,\,
\phi \in [0, {\rm arg} (f(R)+ia_{+}(R)))
\]
and
\[
\Gamma_{R}^{-}=\{\xi|\,\xi=\varrho_{-}e^{i\phi}\},\,\,\,
\phi \in [\pi, {\rm arg} (f(-R)+ia_{+}(-R))).
\]
 We get estimates for the integral (\ref{int1}) first. Let $\xi=\varrho e^{i\phi}$ then
\[
(-\lambda_{-})^{\nu}-(\lambda_{-}-i\xi)^{\nu} \sim -\varrho^{\nu}e^{i(-\pi/2+\phi)\nu}
\]
and
\[
\lambda_{+}^{\nu}-(\lambda_{+}+i\xi)^{\nu} \sim -\varrho^{\nu}e^{i(\pi/2+\phi)\nu}
\]
as $\varrho \rightarrow \infty$.
Since $c_{+}=c_{-}=c>0$ then
\[
c\Gamma(-\nu)\left((-\lambda_{-})^{\nu}-(\lambda_{-}-i\xi)^{\nu}\right)
+c\Gamma(-\nu)\left(\lambda_{+}^{\nu}-(\lambda_{+}+i\xi)^{\nu}\right)
\]
\[
\sim -c\Gamma(-\nu)\rho^{\nu}e^{i\phi \nu}\left(e^{-i\pi \nu /2}+e^{i\pi \nu/2}\right)
=-2c\Gamma(-\nu)\rho^{\nu}e^{i\phi \nu}\cos\left(\frac{\pi \nu}{2}\right)
\]
Hence
\begin{equation} \label{asymp1}
\Re \psi_{\ast}(\xi) \sim  
\varrho \mu\sin \phi +
\varrho^{\nu} 2c\left(-\Gamma(-\nu)\cos\left(\frac{\pi \nu}{2}\right)\right)\cos (\phi \nu).
\end{equation}
Applying (\ref{asymp1}) we get
\[
I_{+}(R)=
\left|\int_{\phi \in [0, {\rm arg} (f(R)+ia_{+}(R)))}
e^{iy\varrho e^{i\phi}-\tau \psi_{\ast}(\varrho e^{i\phi})}\varrho(-\sin \phi +i\cos \phi)d\phi
\right|
\]
\[
\leq
\varrho \left|
\int_{\phi \in [0, {\rm arg} (f(R)+ia_{+}(R)))}
e^{\Re (iy\varrho e^{i\phi}-\tau \psi_{\ast}(\varrho e^{i\phi}))}
\right|
\]
\begin{equation}  \label{chi}
\leq C
\varrho
\int_{0}^{{\rm arg} (f(R)+ia_{+}(R))}
\chi(y,\varrho,\nu,\phi) d\phi,
\end{equation}
where $C$ is a positive constant and
\[
\chi(y,\varrho,\nu,\phi)
\]
\[
:=
\exp\left(-y\varrho \sin \phi-
\tau \varrho
\mu\sin \phi -
 \varrho^{\nu} 2c \left(-\Gamma(-\nu)\cos\left(\frac{\pi \nu}{2}\right)\right)\cos (\phi \nu)
\right)
\]
Since $[0,{\rm arg}(R+ia_{+}(R))) \subset [0, \pi/2)$ and $y \geq 0$ then
\begin{equation} \label{chi1}
\chi(y,\varrho,\nu,\phi)
\leq
\exp\left(
-2c\tau\varrho^{\nu}\left(-\Gamma(-\nu)\cos\left(\frac{\pi \nu}{2}\right)\right)\cos (\phi \nu)\right)
\end{equation}
and $-\Gamma(-\nu)\cos(\pi \nu/2)>0$, $\nu \in (0,2)$. Consequently, from (\ref{chi}) and
(\ref{chi1}) we get
\[
I_{+}(R) \leq
C\varrho \int_{0}^{{\rm arg}(f(R)+ia_{+}(R))}
e^{2c\tau \varrho^{\nu}\Gamma(-\nu)\cos(\pi \nu/2)(1-2\phi\nu /\pi)}d\phi
\]
\[
\leq
C \varrho
e^{2ctau \varrho^{\nu}\Gamma(-\nu)\cos(\pi \nu/2)}
\int_{0}^{\pi/2}
e^{2c\tau \varrho^{\nu}\Gamma(-\nu)\cos(\pi \nu/2)(-2\phi\nu/\pi)}d\phi
\]
\[
\leq
\frac{C\varrho^{1-\nu}}{2c\tau \Gamma(-\nu)\cos (\pi \nu/2)}
e^{2c\tau \varrho^{\nu} \Gamma(-\nu) \cos(\pi \nu /2)(1-\nu)},
\]
where we used the fact that $\cos (\phi \nu) \geq 1-2\phi \nu/\pi$.
From the last line it is easy to see that $\lim_{R \rightarrow \infty} I_{+}(R)=0$ if $\nu \in (0,1)$. Let us get an upper bound for the integral (\ref{int2}).
Assume that $\nu \in (0,1/2]$ then $-\Gamma(-\nu)\cos(\pi \nu/2)\cos(\phi \nu)>0$
for any $\phi \in [\pi/2,\pi]$. Hence
$
\chi(y,\varrho,\nu,\phi)
\leq
\exp\left(
-\tau\varrho \mu\sin \phi\right)
$
and
\[
I_{-}(R) \leq \int_{\phi \in [\pi/2,\pi]}
e^{-\tau \varrho \mu \sin \phi}d\phi
\]
\[
\leq
\int_{\phi \in [\pi/2,\pi]}
e^{-\tau \varrho^{\nu} \mu (-2\phi/\pi+2)}d\phi
<\frac{\pi}{2\tau \varrho \mu}.
\]
Consequently $\lim_{R \rightarrow \infty}\,I_{-}(R)=0$.
\end{proof}Applying Theorem \ref{positive}, Corollary \ref{corollary2} and (%
\ref{integral2}) we get the following form of the approximant $%
F_{call}^{\ast }$ for $F_{call}$, 
\begin{equation*}
F_{call}^{\ast }=e^{-rT}K\int_{\ln (K/S_{0})}^{\infty }\left( e^{\ln
(S_{0}/K)+y}-1\right) p_{T}^{\ast }(y)dy
\end{equation*}%
\begin{equation*}
=e^{-rT}K\int_{\ln (K/S_{0})}^{\infty }\left( e^{\ln (S_{0}/K)+y}-1\right) 
\frac{e^{-\alpha _{+}y}\chi _{\lbrack -\sigma ,\sigma ]}\left( y\right) }{%
2\sigma }
\end{equation*}%
\begin{equation*}
\times \left( \sum_{k\in \mathbb{Z}}e^{-T\psi \left( f\left( \pi k/\sigma
\right) +i\alpha _{+}+ia_{+}\left( \pi k/\sigma \right) \right) }\left( \dot{%
f}\left( \pi k/\sigma \right) +ia_{+}\left( \pi k/\sigma \right) \right)
\right.
\end{equation*}%
\begin{equation*}
\left. e^{iy\left( f\left( \pi k/\sigma \right) -\pi k/\sigma \right)
-ya_{+}\left( \pi k/\sigma \right) +i\pi ky/\sigma }\right) dy:=I_{1}+I_{2},
\end{equation*}%
where%
\begin{equation*}
I_{1}:=\frac{e^{-rT}S_{0}}{2\sigma }\sum_{k\in \mathbb{Z}}\frac{e^{-T\psi
\left( f\left( \pi k/\sigma \right) +i\alpha _{+}+ia_{+}\left( \pi k/\sigma
\right) \right) }\left( \dot{f}\left( \pi k/\sigma \right) +i\dot{a}%
_{+}\left( \pi k/\sigma \right) \right) }{1-\alpha _{+}-a_{+}\left( \pi
k/\sigma \right) +if\left( \pi k/\sigma \right) }
\end{equation*}%
\begin{equation*}
\times \left( e^{\sigma \left( 1-\alpha _{+}-a_{+}\left( \pi k/\sigma
\right) +if\left( \pi k/\sigma \right) \right) }-\left( \frac{K}{S_{0}}%
\right) ^{1-\alpha _{+}-a_{+}\left( \pi k/\sigma \right) +if\left( \pi
k/\sigma \right) }\right)
\end{equation*}%
\begin{equation*}
I_{2}:=-\frac{e^{-rT}K}{2\sigma }\sum_{k\in \mathbb{Z}}\frac{e^{-T\psi
\left( f\left( \pi k/\sigma \right) +i\alpha _{+}+ia_{+}\left( \pi k/\sigma
\right) \right) }\left( \dot{f}\left( \pi k/\sigma \right) +i\dot{a}%
_{+}\left( \pi k/\sigma \right) \right) }{-\alpha _{+}-a_{+}\left( \pi
k/\sigma \right) +if\left( \pi k/\sigma \right) }
\end{equation*}%
\begin{equation*}
\times \left( e^{\sigma \left( -\alpha _{+}-a_{+}\left( \pi k/\sigma \right)
+if\left( \pi k/\sigma \right) \right) }-\left( \frac{K}{S_{0}}\right)
^{1-\alpha _{+}-a_{+}\left( \pi k/\sigma \right) +if\left( \pi k/\sigma
\right) }\right) .
\end{equation*}

\begin{example}
Remind that $\lambda _{+}(\theta )=f(\theta )+i\alpha _{+}+ia_{+}(\theta )$.
Let $f(\theta )=\theta $, $a_{+}(\theta )\equiv 0$, then $\lambda _{+}\left(
\theta \right) =\theta +i\alpha _{+}$. 

In this case $\dot{f}(\theta )+i\dot{a}_{+}(\theta )=1$, 
\begin{equation*}
I_{1}=\frac{e^{-rT}S_{0}}{2\sigma }\sum_{k\in \mathbb{Z}}\frac{e^{-T\psi
\left( \pi k/\sigma +i\alpha _{+}\right) }\left( \left( -1\right)
^{k}e^{\left( 1-\alpha _{+}\right) \sigma }-\left( K/S_{0}\right) ^{1-\alpha
_{+}+i\pi k/\sigma }\right) }{1-\alpha _{+}+i\pi k/\sigma },
\end{equation*}%
\begin{equation*}
I_{2}=-\frac{e^{-rT}S_{0}}{2\sigma }\sum_{k\in \mathbb{Z}}\frac{e^{-T\psi
\left( \pi k/\sigma +i\alpha _{+}\right) }\left( \left( -1\right)
^{k}e^{\left( 1-\alpha _{+}\right) \sigma }-\left( K/S_{0}\right) ^{-\alpha
_{+}+i\pi k/\sigma }\right) }{-\alpha _{+}+i\pi k/\sigma }
\end{equation*}
and $F_{call}^{\ast }=I_{1}+I_{2}.$
\end{example}

\begin{example}
Let $f(\theta )=\theta $, $a_{+}(\theta )=\theta ^{2}$ then $\lambda
_{+}=\theta +i\alpha _{+}+i\theta ^{2}$. In this case $\dot{\lambda}%
_{+}(\theta )=1+2i\theta $ and the contour of integration is a parabola.

It is easy to check that $F_{call}^{\ast }=I_{1}+I_{2}$, where%
\begin{equation*}
I_{1}=\frac{e^{-rT}S_{0}}{2\sigma }\sum_{k\in \mathbb{Z}}\frac{e^{-T\psi
\left( \pi k/\sigma +i\alpha _{+}+i\left( \pi k/\sigma \right) ^{2}\right)
}\left( 1+2i\pi k/\sigma \right) }{1-\alpha _{+}-\left( \pi k/\sigma \right)
^{2}+i\pi k/\sigma }
\end{equation*}%
\begin{equation*}
\times \left( \left( -1\right) ^{k}e^{\left( 1-\alpha _{+}\right) \sigma
-\pi ^{2}k^{2}/\sigma }-\left( K/S_{0}\right) ^{1-\alpha _{+}-\pi
^{2}k^{2}/\sigma ^{2}+i\pi k/\sigma }\right) ,
\end{equation*}%
\begin{equation*}
I_{2}=-\frac{e^{-rT}K}{2\sigma }\sum_{k\in \mathbb{Z}}\frac{e^{-T\psi \left(
\pi k/\sigma +i\alpha _{+}+i\left( \pi k/\sigma \right) ^{2}\right) }\left(
1+2i\pi k/\sigma \right) }{-\alpha _{+}-\left( \pi k/\sigma \right)
^{2}+i\pi k/\sigma }.
\end{equation*}%
\begin{equation*}
\left( \left( -1\right) ^{k}e^{-\alpha _{+}\sigma -\pi ^{2}k^{2}/\sigma
}-\left( K/S_{0}\right) ^{-\alpha _{+}-\pi ^{2}k^{2}/\sigma ^{2}+i\pi
k/\sigma }\right) .
\end{equation*}
\end{example}

\begin{example}
Let $f(\theta )=\theta $, $a_{+}(\theta )=\cosh (\theta ^{2})$, then $%
\lambda \left( \theta \right) =\theta +\alpha _{+}+\cosh (\theta ^{2}),$ $%
\dot{\lambda}_{+}(\theta )=1+2i\theta \sinh (\theta ^{2})$.%
\begin{equation*}
I_{1}=\frac{e^{-rT}S_{0}}{2\sigma }\sum_{k\in \mathbb{Z}}\frac{e^{-T\psi
\left( \pi k/\sigma +i\alpha _{+}+i\cosh \left( \left( \pi k/\sigma \right)
^{2}\right) \right) }\left( 1+\left( 2i\pi k/\sigma \right) \sinh \left(
\left( \pi k/\sigma \right) ^{2}\right) \right) }{1-\alpha _{+}-\cosh \left(
\left( \pi k/\sigma \right) ^{2}\right) +i\pi k/\sigma }
\end{equation*}%
\begin{equation*}
\times \left( \left( -1\right) ^{k}e^{\left( 1-\alpha _{+}\right) \sigma
-\sigma \cosh \left( \left( \pi k/\sigma \right) ^{2}\right) }-\left(
K/S_{0}\right) ^{1-\alpha _{+}-\cosh \left( \left( \pi k/\sigma \right)
^{2}\right) +i\pi k/\sigma }\right) ,
\end{equation*}%
\begin{equation*}
I_{2}=-\frac{e^{-rT}K}{2\sigma }\sum_{k\in \mathbb{Z}}\frac{e^{-T\psi \left(
\pi k/\sigma +i\alpha _{+}+i\cosh \left( \left( \pi k/\sigma \right)
^{2}\right) \right) }\left( 1+\left( 2i\pi k/\sigma \right) \sinh \left(
\left( \pi k/\sigma \right) ^{2}\right) \right) }{-\alpha _{+}-\cosh \left(
\left( \pi k/\sigma \right) ^{2}\right) +i\pi k/\sigma }.
\end{equation*}%
\begin{equation*}
\left( \left( -1\right) ^{k}e^{-\alpha _{+}\sigma -\sigma \cosh \left(
\left( \pi k/\sigma \right) ^{2}\right) }-\left( K/S_{0}\right) ^{-\alpha
_{+}-\cosh \left( \left( \pi k/\sigma \right) ^{2}\right) +i\pi k/\sigma
}\right) .
\end{equation*}
\end{example}

\section{\protect\bigskip Numerical Examples}

\bigskip

In this section we consider several numerical examples. Consider KoBoL
exponent%
\begin{equation*}
\psi \left( \xi \right) =-i\mu \xi +\Gamma \left( -\nu \right) \left[
c_{+}\left( \left( -\lambda _{-}\right) ^{\nu }-\left( -\lambda _{-}-i\xi
\right) ^{\nu }\right) +c_{-}\left( \lambda _{+}^{\nu }-\left( \lambda
_{+}+i\xi \right) ^{\nu }\right) \right] ,
\end{equation*}%
where $\xi =x+iy\in \mathbb{C}$ and $\nu \in \left( 0,2\right) \setminus
\left\{ 1\right\} .$ Observe that $\psi \left( \xi \right) $ is analytic in 
\begin{equation*}
\Omega :=\left\{ z\left\vert z\in \mathbb{C\setminus }\left\{ \left[
i\lambda _{+},\infty \right) \cup \left[ i\lambda _{-},-\infty \right)
\right\} \right. \right\} .
\end{equation*}%
Let us fix European call option parameters. Put $r=0.1,T=0.5,S_{0}=100.$
Assume that KoBoL parameters are $\nu =0.5,c_{+}=c_{-}=1,\lambda
_{+}=5,\lambda _{-}=-5$ \cite{s2}. To satisfy equivalent martingale measure
condition we put $\mu =0.019721$ \cite{levendorskii 2011}. In this case we
have%
\begin{equation*}
\psi \left( \xi \right) =-i0.019721\left( x+iy\right) 
\end{equation*}%
\begin{equation*}
+\Gamma \left( -0.5\right) \left( 2\times 5^{0.5}-\left( 5-i\left(
x+iy\right) \right) ^{0.5}-\left( 5+i\left( x+iy\right) \right)
^{0.5}\right) 
\end{equation*}%
and $\psi \left( \xi \right) $ is analytic in the strip $\left\vert \func{Re}%
\xi \right\vert \leq 5.$  Numerical examples suggest that $a_{+}=1+\left(
\lambda _{+}+1\right) /3.$ In our case $\lambda _{+}=5,$ hence $a_{+}=3.$
Observe that%
\begin{equation*}
M:=\max \left\vert \func{Re}\left( e^{-0.5\left( -i0.019721\left( i5\right)
+\Gamma \left( -0.5\right) \left( 2\times 5^{0.5}-\left( 5-i\left( i5\right)
\right) ^{0.5}-\left( 5+i\left( i5\right) \right) ^{0.5}\right) \right)
}\right) \right\vert 
\end{equation*}%
\begin{equation*}
=\allowbreak 9.\,\allowbreak 702\,279\,703.
\end{equation*}%
Let $A_{\infty ,\delta }UM$ be the set of functions $f\left( z\right) $
which are analytic in the strip $\func{Im}z<\delta $ and such that $\max
\left\{ \left\vert \func{Re}f\left( z\right) \right\vert \left\vert \func{Im}%
z\leq \delta \right. \right\} \leq M.$ For any $f\in A_{\infty ,\delta }UM$
we have the following inequality%
\begin{equation*}
A_{\sigma }\left( f\right) \leq \frac{4M}{\pi }\sum_{k=0}^{\infty }\frac{%
\left( -1\right) ^{k}}{\left( 2k+1\right) \cosh \left( \left( 2k+1\right)
\sigma \delta \right) }<\frac{4M}{\pi }e^{-\delta \sigma },
\end{equation*}%
where%
\begin{equation*}
A_{\sigma }\left( f\right) :=\inf \left\{ \left\Vert f-g\right\Vert _{\infty
}\left\vert g\in W_{\sigma }\right. \right\} 
\end{equation*}%
is the best approximation of $f$ by the subspace $W_{\sigma }$ (see \cite%
{timan} ). In our case $\delta =\lambda _{+}-a_{+}=5-3=2.$ Fix the error of
approximation $\varepsilon >0.$ Then 
\begin{equation*}
A_{\sigma }\left( e^{-T\psi \left( x+a_{+}i\right) }\right) <\frac{4M}{\pi }%
e^{-\delta \sigma }=\varepsilon ,
\end{equation*}%
or solving for $\sigma $ we get%
\begin{equation*}
\sigma =\delta ^{-1}\ln \left( \frac{4M}{\pi \varepsilon }\right) .
\end{equation*}%
Recall that $\sigma $ determines density of interpolation points in
Whittaker-Kotel'nikov-Shannon formula. It means that the step parameter $h$
is 
\begin{equation*}
h=\frac{\pi }{\sigma }.
\end{equation*}

\begin{center}
\begin{tabular}{|l|l|l|}
\hline
$\varepsilon $ & \textbf{Wiener space parameter }$\sigma =\delta ^{-1}\ln
\left( \frac{4M}{\pi \varepsilon }\right) $ & $h=\frac{\pi }{\sigma }$ \\ 
\hline
$E-3$ & $4.\,\allowbreak 710\,840\,317$ & $0.666\,885\,829\,\allowbreak 7$
\\ \hline
$E-4$ & $5.\,\allowbreak 862\,132\,863$ & $0.535\,912\,905\,\allowbreak 3$
\\ \hline
$E-5$ & $7.\,\allowbreak 013\,425\,410$ & $0.447\,939\,839\,\allowbreak 7$
\\ \hline
$E-6$ & $8.\,\allowbreak 164\,717\,956$ & $0.384\,776\,629\,\allowbreak 2$
\\ \hline
$E-7$ & $9.\,\allowbreak 316\,010\,503$ & $0.337\,225\,108\,\allowbreak 6$
\\ \hline
$E-8$ & $10.\,\allowbreak 467\,303\,05$ & $0.300\,133\,915\,\allowbreak 9$
\\ \hline
$E-9$ & $11.\,\allowbreak 618\,595\,60$ & $0.270\,393\,493\,\allowbreak 5$
\\ \hline
$E-10$ & $12.\,\allowbreak 769\,888\,14$ & $0.246\,015\,675\,\allowbreak 2$
\\ \hline
\end{tabular}
\end{center}

\begin{equation*}
\psi \left( \xi \right) =-i0.019721\left( x+iy\right)
\end{equation*}
\begin{equation*}
+\Gamma \left( -0.5\right) \left( 2\times 5^{0.5}-\left( 5-i\left(
x+iy\right) \right) ^{0.5}-\left( 5+i\left( x+iy\right) \right) ^{0.5}\right)
\end{equation*}

We approximate the density function 
\begin{equation*}
p_{T}\left( x\right) =\frac{1}{2\pi }e^{-a_{+}x}\int_{\mathbb{R}%
}e^{ixy}e^{-T\psi \left( y+a_{+}i\right) }dy
\end{equation*}%
by 
\begin{equation*}
p_{T}^{\ast }\left( x\right) =\frac{1}{2\pi }e^{-a_{+}x}\int_{\mathbb{R}%
}e^{ixy}g\left( y\right) dy,
\end{equation*}%
where $g\left( y\right) \in W_{\sigma }$ interpolates the function 
\begin{equation*}
f\left( y\right) :=\left\{ 
\begin{array}{cc}
e^{-T\psi \left( y+a_{+}i\right) }, & y\in \left[ -A,A\right] , \\ 
0, & y\in \mathbb{R\setminus }\left[ -A,A\right] ,%
\end{array}%
\right.
\end{equation*}%
at the points $x_{k}=k\pi /\sigma ,k\in \mathbb{Z}$ and $A>0$ is fixed.
Clearly, 
\begin{equation*}
\left\vert p_{T}\left( x\right) -p_{T}^{\ast }\left( x\right) \right\vert =%
\frac{1}{2\pi }\left\vert \int_{\mathbb{R}}e^{ixy}\left( e^{-T\psi \left(
y\right) }-g\left( y\right) \right) dy\right\vert
\end{equation*}%
\begin{equation*}
\leq \frac{1}{2\pi }\int_{[-A,A]}\left\vert e^{-T\psi \left( y\right)
}-g\left( y\right) \right\vert dy+\frac{1}{2\pi }\left\vert \int_{\mathbb{%
R\setminus \lbrack }-A,A]}e^{-T\psi \left( y\right) }dy\right\vert .
\end{equation*}%
Since $\left\vert e^{-T\psi \left( y\right) }-g\left( y\right) \right\vert
\leq \varepsilon $ for any $y\in \mathbb{R}$ then for the first integral we
get 
\begin{equation*}
\frac{1}{2\pi }\int_{[-A,A]}\left\vert e^{-T\psi \left( y\right) }-g\left(
y\right) \right\vert dy\leq \frac{A\varepsilon }{\pi }.
\end{equation*}%
To estimate the second integral we assume as before $\nu
=0.5,c_{+}=c_{-}=1,\lambda _{+}=5,\lambda _{-}=-5$ and $T=0.5.$ Calculations
show the following result.

\begin{center}
\bigskip 
\begin{tabular}{|l|l|}
\hline
\textbf{Truncation parameter }$A$ & $\mathbf{\varepsilon }^{\ast }:=\frac{1}{%
2\pi }\left\vert \int_{\mathbb{R\setminus \lbrack }-A,A]}e^{-T\psi \left(
y\right) }dy\right\vert $ \\ \hline
$10$ & $6.\,\allowbreak 626\,537\,364\times 10^{-2}$ \\ \hline
$20$ & $5.\,\allowbreak 781\,601\,106\times 10^{-3}$ \\ \hline
$30$ & $7.\,\allowbreak 180\,593\,247\times 10^{-4}$ \\ \hline
$40$ & $1.\,\allowbreak 138\,385\,230\times 10^{-4}$ \\ \hline
$50$ & $2.\,\allowbreak 153\,105\,090\times 10^{-5}$ \\ \hline
$60$ & $4.\,\allowbreak 657\,594\,108\times 10^{-6}$ \\ \hline
$70$ & $1.\,\allowbreak 120\,585\,522\times 10^{-6}$ \\ \hline
$80$ & $\allowbreak 2.\,\allowbreak 940\,645\,917\times 10^{-7}$ \\ \hline
$90$ & $8.\,\allowbreak 298\,093\,791\times 10^{-8}$ \\ \hline
$100$ & $2.\,\allowbreak 491\,098\,701\times 10^{-8}$ \\ \hline
$110$ & $7.\,\allowbreak 889\,750\,755\times 10^{-9}$ \\ \hline
$120$ & $2.\,\allowbreak 618\,921\,335\times 10^{-9}$ \\ \hline
$130$ & $\allowbreak 9.\,\allowbreak 062\,377\,049\times 10^{-10}$ \\ \hline
\end{tabular}
\end{center}

The total error of approximation of density function is 
\begin{equation*}
\epsilon :=\frac{A\varepsilon }{\pi }+\varepsilon ^{\ast }
\end{equation*}%
Let, in particular, $\varepsilon =10^{-7}$ $\left( \sigma =9.\,\allowbreak
316\,010\,503\right) .$ Then selecting $A=50$ we get $\varepsilon ^{\ast
}=2.\,\allowbreak 153\,105\,090\times 10^{-5}$ and $\epsilon =50\times
10^{-7}\pi ^{-1}+2.\,\allowbreak 153\,105\,090\times 10^{-5}=\allowbreak
2.\,\allowbreak 312\,260\,033\times 10^{-5}.$ We should take $N$ terms in
our approximant, where $\pi N/\sigma =A$, or $N=A\sigma /\pi .$ In our case $%
N=50\times 9.\,\allowbreak 316\,010\,503/\pi =149.$

\end{document}